\documentclass[12pt]{article} \usepackage{amsmath,amssymb,amsthm,mathtools} \usepackage{aliascnt} \usepackage[margin=1in]{geometry} \usepackage{hyperref} \usepackage[nameinlink,noabbrev]{cleveref} \usepackage{enumitem}

\newcommand{\bbP}{\mathbb{P}} \newcommand{\bbE}{\mathbb{E}} \newcommand{\bbR}{\mathbb{R}} \newcommand{\parent}[1]{\left(#1\right)} \newcommand{\defi}[1]{\emph{#1}} \DeclareMathOperator{\Cov}{Cov} \DeclareMathOperator{\Var}{Var} \DeclareMathOperator{\rank}{rank} \DeclareMathOperator{\pmi}{pmi} \DeclareMathOperator{\logit}{logit} \DeclareMathOperator{\sgn}{sgn} \DeclareMathOperator{\OR}{OR}

\crefname{section}{section}{sections} \Crefname{section}{Section}{Sections} \crefname{subsection}{subsection}{subsections} \Crefname{subsection}{Subsection}{Subsections} \crefname{theorem}{theorem}{theorems} \Crefname{theorem}{Theorem}{Theorems} \crefname{proposition}{proposition}{propositions} \Crefname{proposition}{Proposition}{Propositions} \crefname{lemma}{lemma}{lemmas} \Crefname{lemma}{Lemma}{Lemmas} \crefname{corollary}{corollary}{corollaries} \Crefname{corollary}{Corollary}{Corollaries} \crefname{definition}{definition}{definitions} \Crefname{definition}{Definition}{Definitions} \crefname{example}{example}{examples} \Crefname{example}{Example}{Examples} \crefname{remark}{remark}{remarks} \Crefname{remark}{Remark}{Remarks} \crefname{equation}{equation}{equations} \Crefname{equation}{Equation}{Equations}

\newtheorem{theorem}{Theorem}[section]

\newaliascnt{proposition}{theorem} \newtheorem{proposition}[proposition]{Proposition} \aliascntresetthe{proposition}

\newaliascnt{lemma}{theorem} \newtheorem{lemma}[lemma]{Lemma} \aliascntresetthe{lemma}

\newaliascnt{corollary}{theorem}  \aliascntresetthe{corollary}

\theoremstyle{definition} \newaliascnt{definition}{theorem} \newtheorem{definition}[definition]{Definition} \aliascntresetthe{definition}

\newaliascnt{example}{theorem} \newtheorem{example}[example]{Example} \aliascntresetthe{example}

\theoremstyle{remark} \newaliascnt{remark}{theorem} \newtheorem{remark}[remark]{Remark} \aliascntresetthe{remark}

\title{On the Symmetry of Evidential Support} \author{Grant Molnar} \date{\today}

\begin{document} \maketitle

\begin{abstract}
For events $A$ and $B$, we have
\[
\mathbb{P}(A\mid B) > \mathbb{P}(A\mid \neg B) \qquad\Longleftrightarrow\qquad \mathbb{P}(B\mid A) > \mathbb{P}(B\mid \neg A)
\]
whenever all four quantities are defined. In other words, $B$ is evidence for $A$ if and only if $A$ is evidence for $B$. This note gives seven different proofs of this fact---by cross-multiplication, covariance, coupling parameters, odds ratios, pointwise mutual information, combinatorial double counting, and mixed discrete derivatives---and develops a surrounding web of interpretations. Once the marginals  $\mathbb{P}(A)$ and $\mathbb{P}(B)$ are fixed, a $2\times 2$ table has only one degree of freedom, so every scalar notion of positive association must be governed by the same signed parameter.
\end{abstract}

\tableofcontents

%% ====================================================================
\section{Introduction}\label{sec:intro}
%% ====================================================================

Throughout this note, $A$ and $B$ are events in an ambient probability space with \[ 0 < \bbP(A) < 1 \qquad\text{and}\qquad 0 < \bbP(B) < 1, \] so all four conditional probabilities $\bbP(A\mid B)$, $\bbP(A\mid \neg B)$, $\bbP(B\mid A)$, $\bbP(B\mid \neg A)$ are well-defined.

Write the four joint probabilities as 
\[
p \coloneqq \bbP(A\wedge B),\quad q \coloneqq \bbP(A\wedge \neg B),\quad r \coloneqq \bbP(\neg A\wedge B),\quad s \coloneqq \bbP(\neg A\wedge \neg B),
\]
so $p+q+r+s=1$. The marginals are $\bbP(A)=p+q$ and $\bbP(B)=p+r$. Define $\Delta \coloneqq ps - qr$.

The conceptual core of the note is that binary association is one-dimensional. Once $\bbP(A)$ and $\bbP(B)$ are fixed, a $2\times 2$ table has only one free parameter, so every scalar measure of positive versus negative association must change sign at the same threshold.

\begin{theorem}[Symmetry of evidential support]\label{thmmain} The following are equivalent. \begin{enumerate}[label=(\roman*)] \item $\bbP(A\mid B) > \bbP(A\mid \neg B)$. \item $\bbP(B\mid A) > \bbP(B\mid \neg A)$. \item $\Delta > 0$. \end{enumerate} Replacing every ``$>$'' with ``$=$'' or with ``$<$'' gives two further equivalences, so the three possible signs of $\Delta$ classify the three possible qualitative regimes of binary association. \end{theorem}

\Cref{thmmain} is classical, and several of the proofs presented here are standard or folkloric; the contribution of the note is to gather them together and organize them around the observation that, once the marginals are fixed, binary association is one-dimensional.

\begin{remark}[Two binary encodings] We will use two standard numerical encodings of the events $A$ and $B$. The default is the \defi{indicator encoding} \[ X \coloneqq \mathbf{1}_A,\qquad Y \coloneqq \mathbf{1}_B, \] so $X,Y\in\{0,1\}$. This is the natural choice for conditional probabilities, covariances of indicators, and regression formulas. In the Walsh--Fourier and Ising-model viewpoints it is more convenient to use the \defi{spin encoding} \[ U \coloneqq 2X-1,\qquad V \coloneqq 2Y-1, \] so $U,V\in\{-1,+1\}$. We will always signal explicitly when we pass from one encoding to the other. \end{remark}

Before we prove \Cref{thmmain}, we give a few examples of how it matters.

\begin{example}[Medical testing] Let $A$ be the event that a patient has a disease and $B$ the event that a screening test is positive. If the test is informative---meaning $\bbP(B\mid A) > \bbP(B\mid \neg A)$, i.e.\ the true-positive rate exceeds the false-positive rate---then \Cref{thmmain} guarantees that a positive result raises the probability of disease relative to a negative result. \end{example}

\begin{example}[Correlation and prediction] Suppose $A$ is the event that a student scores above the median on a mathematics exam and $B$ is the event that the same student scores above the median on a physics exam. If $A$ makes $B$ more likely, then $B$ also makes $A$ more likely. The two predictions ``math predicts physics'' and ``physics predicts math'' stand or fall together. \end{example}

\begin{example}[Paired performance tests] Let $A$ be the event that a manufactured component passes a thermal-stress test and $B$ the event that it passes a vibration test. If passing the thermal test makes passing the vibration test more likely, then passing the vibration test likewise makes passing the thermal test more likely. The same association governs both predictive directions. \end{example}

\medskip

The remainder of the note presents seven proofs organized by viewpoint. \Cref{sec:alg} begins with the direct algebraic reduction to the scalar invariant $\Delta=ps-qr$. \Cref{sec:cov} recasts the sign condition in terms of covariance and, more conceptually, in the language of Hilbert spaces. \Cref{sec:coupling} then supplies a global structural explanation: once the marginals are fixed, there is only one interaction coordinate. \Cref{sec:odds} interprets the phenomenon statistically through odds ratios and Bayes factors, while \Cref{sec:pmi} gives an information-theoretic formulation in terms of pointwise mutual information and surprisal reduction. Finally, \Cref{sec:comb} recasts the same sign condition combinatorially, and \Cref{sec:logpot} interprets it using discrete potential theory. \Cref{sec:conclusion} organizes the previous results and emphasizes why, in the $2\times 2$ case, so many scalar notions of association collapse to the same sign.

%% ====================================================================
\section{The Algebraic Proof}\label{sec:alg}
%% ====================================================================

\subsection{Perspective}

This section records our most direct proof of \Cref{thmmain}. Both conditional-probability comparisons reduce by elementary algebra to the single scalar inequality $ps-qr>0$, so the determinant-like quantity $\Delta=ps-qr$ emerges immediately as the basic invariant of the $2\times 2$ table.

\subsection{Proof}

\begin{proof}[Proof of \Cref{thmmain}, algebraic version] We show that both (i) and (ii) are equivalent to (iii).

For (i)$\,\Leftrightarrow\,$(iii), write \[ \bbP(A\mid B) > \bbP(A\mid \neg B) \iff \frac{p}{p+r} > \frac{q}{q+s}. \] Both denominators are positive, so cross-multiplying preserves the inequality, yielding
\[ p(q+s) > q(p+r) \iff ps > qr \iff \Delta > 0. \]

For (ii)$\,\Leftrightarrow\,$(iii), write \[ \bbP(B\mid A) > \bbP(B\mid \neg A) \iff \frac{p}{p+q} > \frac{r}{r+s}. \] Cross-multiplying gives \[ p(r+s) > r(p+q) \iff ps > qr \iff \Delta > 0. \] Our claim follows. \end{proof}

\subsection{The determinant and rank one}

Arrange the joint probabilities in the matrix \[ M \coloneqq \begin{pmatrix} p & q \\ r & s \end{pmatrix}. \] Then $\det M = ps - qr = \Delta$. The algebraic proof can now be reread as a statement about $\det M$.

\begin{proposition}\label{proprank} $A$ and $B$ are independent if and only if $M$ has rank~$1$. \end{proposition}

\begin{proof} If $A$ and $B$ are independent, then $M$ factors as an outer product of marginal vectors, \[ M = \begin{pmatrix} \bbP(A) \\ \bbP(\neg A) \end{pmatrix} \begin{pmatrix} \bbP(B) & \bbP(\neg B) \end{pmatrix}. \] Any outer product has rank at most~$1$, and since $M\neq 0$ the rank is exactly~$1$.

Conversely, if $\rank\parent{M}=1$, then the second row is a scalar multiple of the first: $(r,s)=\lambda\,\parent{p,q}$ for some $\lambda \ge 0$. Hence \[ \bbP(B)=p+r=(1+\lambda)p, \qquad \bbP(A)=p+q. \] Also \[ 1=p+q+r+s=(1+\lambda)(p+q), \] so $(1+\lambda)(p+q)=1$. Therefore \[ \bbP(A)\bbP(B)=(p+q)(1+\lambda)p=p=\bbP(A\wedge B). \] Thus $A$ and $B$ are independent. Since $M$ is $2\times 2$ with nonnegative entries summing to~$1$, its rank is $1$ or $2$, so rank~$1$ is equivalent to $\det M = 0$. \end{proof}

Consider the conditional-probability differences
\begin{align}
\bbP(A\mid B) - \bbP(A\mid \neg B) &= \frac{\det M}{(p+r)(q+s)}, \label{eqdetrow} \\
[4pt] \bbP(B\mid A) - \bbP(B\mid \neg A) &= \frac{\det M}{(p+q)(r+s)}, \label{eqdetcol} 
\end{align}
with strictly positive denominators and the same numerator. This is because \eqref{eqdetcol} is \eqref{eqdetrow} applied to $M^\top$, and $\det M^\top = \det M$.

\begin{remark}[Algebraic statistics] In the language of algebraic statistics \cite{DrtonSturmfelsSullivant}, the independence model for a $2\times 2$ contingency table is the variety $ps - qr = 0$; this is the simplest instance of the Segre variety. For an $m\times n$ table, independence is equivalent to the vanishing of all $2\times 2$ minors. The binary case is special because there is only one minor. \end{remark}

\subsection{Signed area and geometry}

The two rows of $M$ are vectors in $\bbR^2$, \[ \mathbf{r}_A \coloneqq (p,q), \qquad \mathbf{r}_{\neg A} \coloneqq (r,s), \] and $\Delta$ is the signed area of the parallelogram they span. When $\Delta = 0$ the rows are proportional---the conditional distribution of $B$ given $A$ equals that given $\neg A$---which is independence. When $\Delta > 0$, the vector $\mathbf{r}_A$ is ``tilted toward the first coordinate'' relative to $\mathbf{r}_{\neg A}$, meaning $A$ is more concentrated on $B$ than $\neg A$ is. Transposing $M$ preserves the determinant, so the same tilt governs the column comparison. In exterior algebra, $\Delta = \mathbf{r}_A \wedge \mathbf{r}_{\neg A}$, and \Cref{thmmain} follows because $\det M^\top = \det M$.

%% ====================================================================
\section{The Covariance Proof}\label{sec:cov}
%% ====================================================================

\subsection{Perspective}

This section recasts \Cref{thmmain} in the language of Hilbert spaces. The key point is that both conditional directions are controlled by the sign of the single bilinear quantity $\Cov(\mathbf{1}_A,\mathbf{1}_B)$, or equivalently by the associated one-dimensional correlation parameter.

\subsection{Proof}

Let $X = \mathbf{1}_A$ and $Y = \mathbf{1}_B$ be as usual.

\begin{proposition}\label{propcoveq} We have the equality $\Cov\parent{X,Y} = \Delta$. \end{proposition}

\begin{proof} We compute \[ \Cov\parent{X,Y} = \bbE[XY] - \bbE[X]\,\bbE[Y] = p - (p+q)(p+r) = p(1-p-q-r) - qr = ps - qr. \] \end{proof}

\begin{proposition}\label{propcovcond} For any events $A$ and $B$ with $0 < \bbP(B) < 1$, \[ \bbP(A\mid B) - \bbP(A) = \frac{\Cov\parent{X,Y}}{\bbP(B)}. \] \end{proposition}

\begin{proof} We compute \[ \bbP(A\mid B) - \bbP(A) = \frac{\bbP(A\wedge B)}{\bbP(B)} - \bbP(A) = \frac{\bbP(A\wedge B) - \bbP(A)\,\bbP(B)}{\bbP(B)} = \frac{\Cov\parent{X,Y}}{\bbP(B)}. \] \end{proof}

\begin{proof}[Proof of \Cref{thmmain}, covariance version] By \Cref{propcovcond}, $\bbP(A\mid B) > \bbP(A)$ if and only if $\Cov\parent{X,Y} > 0$. Applying the same identity with $\neg B$ in place of $B$ and noting $\Cov(X, \mathbf{1}_{\neg B}) = -\Cov\parent{X,Y}$ gives $\bbP(A\mid \neg B) < \bbP(A)$ if and only if $\Cov\parent{X,Y} > 0$. Combining these identities, we conclude \[ \bbP(A\mid B) > \bbP(A\mid \neg B) \iff \Cov\parent{X,Y} > 0. \] Since $\Cov\parent{X,Y} = \Cov\parent{Y,X}$, the same equivalence holds with $A$ and $B$ swapped. \end{proof}

A useful identity that emerges is \[ \bbP(A\mid B) - \bbP(A\mid \neg B) = \Cov\parent{X,Y} \left(\frac{1}{\bbP(B)} + \frac{1}{\bbP(\neg B)}\right). \]

\subsection{Regression slope}

The slope of the least-squares regression of $X$ on $Y$ is $\beta_{X\mid Y} \coloneqq \Cov\parent{X,Y}/\Var\parent{Y}$, whose sign is the sign of $\Cov\parent{X,Y} = \Delta$. The regression of $Y$ on $X$ has slope $\beta_{Y\mid X} \coloneqq \Cov\parent{X,Y}/\Var\parent{X}$, which shares the same sign. So ``the conditional mean of $\mathbf{1}_A$ is increasing in $\mathbf{1}_B$'' and vice versa are the same assertion.

\subsection{Walsh--Fourier expansion}

Encode the events as $\pm 1$ variables $U = 2X-1$ and $V = 2Y-1$. Let
\[
f(u,v)\coloneqq \bbP(U=u,V=v)
\]
be the joint probability mass function on $\{-1,+1\}^2$. With respect to the Walsh basis
\[
1,\quad u,\quad v,\quad uv
\]
for real-valued functions on $\{-1,+1\}^2$, every such $f$ has a unique expansion
\[
f(u,v)=\hat f_{\varnothing}+\hat f_U\,u+\hat f_V\,v+\hat f_{UV}\,uv.
\]
We note $\hat f_{UV}=\frac14\bbE[UV]$. Writing $\alpha = \bbP(A)$ and $\beta = \bbP(B)$, we obtain
\[
\bbE[UV]
=
4p-2\alpha-2\beta+1.
=
(2\alpha-1)(2\beta-1)+4\Delta,
\]
which is an affine function of $\Cov(X, Y) = \Delta$ with marginals held fixed. 

After the marginal terms are separated off, there remains exactly one second-order interaction coefficient, and its sign is the sign of $\Delta$. This perspective is natural in the analysis of Boolean functions and on the discrete cube \cite[Ch.~1]{ODonnell}.{ODonnell}.

\subsection{Hilbert-space and canonical-correlation perspective}

The covariance proof can be deepened into a statement about conditional-expectation operators on $L^2$.

Define the centered, standardized indicators \[ u_A \coloneqq \frac{\mathbf{1}_A - \bbP(A)}{\sqrt{\bbP(A)\,\bbP(\neg A)}}, \qquad u_B \coloneqq \frac{\mathbf{1}_B - \bbP(B)}{\sqrt{\bbP(B)\,\bbP(\neg B)}}. \] Each has mean zero and variance one. The space of $B$-measurable mean-zero random variables is one-dimensional, spanned by $u_B$. Therefore $\bbE[u_A \mid B]$, which is $B$-measurable and mean-zero, must be a scalar multiple of $u_B$. Write \[ \bbE[u_A \mid B] = \rho\, u_B. \] Taking the inner product with $u_B$ gives $\rho = \bbE[u_A\, u_B]$, the canonical correlation coefficient. By exactly the same argument with $A$ and $B$ swapped, we have \[ \bbE[u_B \mid A] = \rho\, u_A. \]

Now $u_B$ is positive on $B$ and negative on $\neg B$. Hence $\rho > 0$ if and only if $\bbE[u_A \mid B] > \bbE[u_A \mid \neg B]$. Since $u_A$ is a positive affine function of $\mathbf{1}_A$, this is equivalent to $\bbP(A\mid B) > \bbP(A\mid \neg B)$. By the symmetric identity $\bbE[u_B \mid A] = \rho\, u_A$, the same condition $\rho > 0$ is equivalent to $\bbP(B\mid A) > \bbP(B\mid \neg A)$.

Thus, the two conditional-expectation operators are adjoints between two one-dimensional mean-zero Hilbert spaces, so there is a single signed scalar---the canonical correlation $\rho$---governing both directions. In the binary case $\rho$ is just the Pearson correlation \[ \phi = \frac{\Delta}{\sqrt{\bbP(A)\,\bbP(\neg A)\,\bbP(B)\,\bbP(\neg B)}}, \] but the Hilbert-space formulation generalizes to arbitrary finite partitions (where the canonical correlations form a vector, and the simple symmetry can fail).

\begin{remark}[Orthogonal projection] The identity $\bbE[u_A \mid B] = \rho\, u_B$ can also be derived by projecting the centered indicator $\mathbf{1}_A - \bbP(A)$ onto the two-dimensional subspace of $L^2(\Omega)$ spanned by $1$ and $\mathbf{1}_B$. The projection coefficient is \[ \frac{\langle \mathbf{1}_A - \bbP(A),\, \mathbf{1}_B - \bbP(B)\rangle}{ \|\mathbf{1}_B - \bbP(B)\|^2} = \frac{\Cov\parent{X,Y}}{\Var\parent{Y}}, \] which is the regression slope $\beta_{X\mid Y}$. Swapping $A$ and $B$ swaps the two centered indicators inside the same symmetric inner product $\langle\cdot,\cdot\rangle$, so the sign is preserved. \end{remark}

%% ====================================================================
\section{The Coupling-Parameter Proof}\label{sec:coupling}
%% ====================================================================

\subsection{Perspective}

This section gives a structural explanation for \Cref{thmmain}. Once the marginals are fixed, a $2\times 2$ table has only one degree of freedom, so both conditional-probability differences must be governed by the same interaction coordinate.

\subsection{Proof}

Let $\alpha \coloneqq \bbP(A)$ and $\beta \coloneqq \bbP(B)$. Every $2\times 2$ law with these marginals is determined by $p = \bbP(A\wedge B)$; the remaining cells are then forced by the marginal constraints, \[ q = \alpha - p,\qquad r = \beta - p,\qquad s = 1 - \alpha - \beta + p. \] Define the \defi{coupling parameter} \[ t \coloneqq p - \alpha\beta, \] which measures the excess of $\bbP(A\wedge B)$ over its independence value $\alpha\beta$. Then $t = 0$ at independence, $t > 0$ for positive association, and $t < 0$ for negative association.

\begin{proof}[Proof of \Cref{thmmain}, coupling-parameter version] Express the two conditional-probability differences in terms of $t$. Since $p + r = \beta$ and $q + s = 1 - \beta$, we have
\[
\bbP(A\mid B) - \bbP(A\mid \neg B) = \frac{p}{\beta} - \frac{q}{1-\beta} = \frac{p(1-\beta) - q\beta}{\beta(1-\beta)}.
\]
Substituting $p = \alpha\beta + t$ and $q = \alpha(1-\beta) - t$ into the numerator gives \[ (\alpha\beta + t)(1-\beta) - (\alpha(1-\beta) - t)\beta = t(1-\beta) + t\beta = t. \] Thus, we have \[ \bbP(A\mid B) - \bbP(A\mid \neg B) = \frac{t}{\beta(1-\beta)}. \] Likewise, since $p + q = \alpha$ and $r + s = 1 - \alpha$, we have \[ \bbP(B\mid A) - \bbP(B\mid \neg A) = \frac{p}{\alpha} - \frac{r}{1-\alpha} = \frac{t}{\alpha(1-\alpha)}. \] Both are quotients of $t$ by a strictly positive denominator, so both have the sign of $t$. \end{proof}

This proof presents \Cref{thmmain} as a statement about the geometry of the $2 \times 2$ contingency table---once marginals are fixed, there is only one degree of freedom, and every directional notion of positive association must flip sign at the independence point \cite{FienbergGilbert1970}. From the discrete-copula viewpoint, the same family may be regarded as the bivariate Bernoulli family with fixed margins and one remaining dependence coordinate \cite{Geenens2020}.

\subsection{Normalized interaction parameter}

We obtain a convenient normalized interaction coordinate by writing
\[
\bbP(A\wedge B)=\alpha\beta+\theta\,\alpha(1-\alpha)\beta(1-\beta)
\]
for a scalar $\theta$. The remaining cells are then forced by the marginal constraints. A direct calculation gives
\[
\bbP(A\mid B)-\bbP(A\mid \neg B)=\theta\,\alpha(1-\alpha),
\qquad
\bbP(B\mid A)-\bbP(B\mid \neg A)=\theta\,\beta(1-\beta).
\]
Hence both conditional-probability differences have the same sign as $\theta$. In the bivariate Bernoulli setting this is a convenient one-parameter reparameterization of the dependence structure; for fixed marginals it differs from the Pearson $\phi$-coefficient only by the positive scale factor $\sqrt{\alpha(1-\alpha)\beta(1-\beta)}$. Indeed, we have
\[
\phi
=
\frac{\Delta}{\sqrt{\alpha(1-\alpha)\beta(1-\beta)}}
=
\theta\,\sqrt{\alpha(1-\alpha)\beta(1-\beta)}
\]
\cite{Geenens2020,Warrens2008}.

\subsection{The differential viewpoint}

The coupling-parameter proof has a natural differential refinement. For each admissible value of $t$, let $\bbP_t$ denote the unique joint law of $(A,B)$ with fixed marginals $\bbP_t(A)=\alpha$ and $\bbP_t(B)=\beta$ and with coupling parameter $t$, i.e. \[ \bbP_t(A\wedge B)=\alpha\beta+t,\qquad \bbP_t(A\wedge \neg B)=\alpha(1-\beta)-t, \] \[ \bbP_t(\neg A\wedge B)=\beta(1-\alpha)-t,\qquad \bbP_t(\neg A\wedge \neg B)=(1-\alpha)(1-\beta)+t. \] Here ``admissible'' means precisely that these four quantities are nonnegative, i.e. that \[ -\min\{\alpha\beta,(1-\alpha)(1-\beta)\} \le t \le \min\{\alpha(1-\beta),(1-\alpha)\beta\}. \] Now define \begin{align*} F(t) &\coloneqq \bbP_t(A\mid B) - \bbP_t(A\mid \neg B) = \frac{t}{\beta(1-\beta)}, \\ G(t) &\coloneqq \bbP_t(B\mid A) - \bbP_t(B\mid \neg A) = \frac{t}{\alpha(1-\alpha)}. \end{align*} Then $F$ and $G$ are both linear in $t$ with strictly positive slopes \[ F'(t) = \frac{1}{\beta(1-\beta)} > 0, \qquad G'(t) = \frac{1}{\alpha(1-\alpha)} > 0, \] and $F(0) = G(0) = 0$. As $F$ and $G$ are strictly increasing through the origin, we have $\sgn F\parent{t} = \sgn t = \sgn G\parent{t}$, and \Cref{thmmain} follows.

\subsection{Concordance and discordance}

Take two independent draws $(X,Y)$ and $(X',Y')$ from the joint law, with $X = \mathbf{1}_A$, $Y = \mathbf{1}_B$. Define a draw pair to be \defi{concordant} if one sample is $(1,1)$ and the other is $(0,0)$, and \defi{discordant} if one is $(1,0)$ and the other is $(0,1)$. Then \[ \bbP(\text{concordant}) = 2ps, \qquad \bbP(\text{discordant}) = 2qr. \] So $\Delta > 0$ if and only if concordant pairs are more likely than discordant pairs. This condition is visibly invariant under swapping the two coordinates, providing yet another proof of \Cref{thmmain}. $B$ is evidence for $A$ if and only if, in two independent draws, aligned outcomes are more probable than crossed outcomes. This is reminiscent of Kendall's $\tau$ \cite{Kendall1938}.

The connection to the coupling parameter is immediate, since \[ \bbE[(X-X')(Y-Y')] = 2\Cov\parent{X,Y} = 2\Delta, \] and the product $(X-X')(Y-Y')$ equals $+1$ on concordant pairs, $-1$ on discordant pairs, and $0$ otherwise.

\subsection{Transport and mismatch cost}

For Bernoulli variables $X = \mathbf{1}_A$ and $Y = \mathbf{1}_B$ with fixed marginals $\alpha$ and $\beta$, the Hamming mismatch probability is 
\begin{equation}\label{eqnXneqY}
    \bbP(X \neq Y) = q + r.
\end{equation}
Under independence, \eqref{eqnXneqY} equals $\alpha(1-\beta) + (1-\alpha)\beta$. Positive association ($\Delta > 0$) is equivalent to the actual coupling having less Hamming mismatch than the independence coupling, i.e., \[ q + r < \alpha(1-\beta) + (1-\alpha)\beta. \] Expanding, this reduces to $p > \alpha\beta$, i.e.\ $t > 0$. So positive evidence is the same as saying the joint law is a better-than-independent coupling in the sense of discrete optimal transport.

%% ====================================================================
\section{The Odds-Ratio Proof}\label{sec:odds}
%% ====================================================================

\subsection{Perspective}

This section interprets \Cref{thmmain} on the odds scale used in statistics and Bayesian inference. The key invariant here is the odds ratio $ps/qr$, whose sign relative to $1$ governs both conditional comparisons symmetrically.

\subsection{Proof}

\begin{proof}[Proof of \Cref{thmmain}, odds-ratio version] The function $\pi \mapsto \pi/(1-\pi)$ on the unit interval $(0, 1)$ is strictly increasing. Thus, $\bbP(A\mid B) > \bbP(A\mid \neg B)$ if and only if the odds of $A$ given $B$ exceed the odds given $\neg B$. The odds of $A$ given $B$ are \[ \frac{\bbP(A\mid B)}{\bbP(\neg A\mid B)} = \frac{p}{r}, \] and the odds given $\neg B$ are $q/s$. The inequality becomes \[ \frac{p}{r} > \frac{q}{s} \iff ps > qr \iff \Delta > 0. \] Likewise, $\bbP(B\mid A) > \bbP(B\mid \neg A)$ translates to $p/q > r/s$, which is again $ps > qr$. Both are equivalent to the odds ratio $\OR \coloneqq ps/qr$ exceeding~$1$ \cite{Agresti}. \end{proof}

\begin{remark}[Historical note] The odds ratio as a measure of association goes back at least to Yule \cite{Yule1900}. In the present section, however, \cite{Agresti} is the more direct technical reference for the equivalence between positive association in a $2\times 2$ table, odds ratios, and logistic coefficients. \end{remark}

\subsection{Bayesian reciprocity of evidence}

Bayes' rule in odds form \cite[Ch.~1]{Good} states that after observing $B$, the posterior odds of $A$ relate to the prior odds by \[ \underbrace{\frac{\bbP(A\mid B)}{\bbP(\neg A\mid B)}}_{\text{posterior odds}} = \underbrace{\frac{\bbP(A)}{\bbP(\neg A)}}_{\text{prior odds}} \;\cdot\; \underbrace{\frac{\bbP(B\mid A)}{\bbP(B\mid \neg A)}}_{\text{Bayes factor}}. \] The Bayes factor exceeds~$1$ if and only if $\bbP(B\mid A) > \bbP(B\mid \neg A)$, meaning observing $B$ raises the odds of $A$. This is equivalent to $\OR > 1$, which is symmetric in $A$ and $B$. The Bayesian reading of \Cref{thmmain} is that evidential support is reciprocal---$B$ is evidence for $A$ if and only if $A$ is evidence for $B$ \cite{EellsFitelson,CrupiTentoriGonzalez}.

\begin{remark} This reciprocity holds for binary partitions. Once one passes to finer partitions or conditional stratifications, related reversals can occur; compare Simpson's paradox \cite{Blyth}. \end{remark}

\subsection{The Ising interaction and exponential families}

The log-odds ratio $\log(ps/qr)$ is the natural interaction parameter of the $2\times 2$ table viewed as a log-linear or exponential-family model. Write the joint mass function in Gibbs form as \[ \bbP(U=u,\, V=v) = \frac{1}{Z}\,\exp\!\bigl(\alpha\, u + \beta\, v + \gamma\, uv\bigr), \qquad u,v\in\{-1,+1\}, \] where $U,V$ are the $\pm 1$ encodings of $A,B$. A direct calculation shows that \[ \gamma = \tfrac14 \log(ps/qr). \] Thus the sign of the unique interaction parameter $\gamma$ is the sign of the log-odds ratio, hence also the sign of $\Delta$. In the language of statistical mechanics, $\gamma>0$ is \textbf{ferromagnetic} (aligned configurations favored), $\gamma<0$ is \textbf{antiferromagnetic}, and $\gamma=0$ is independence. See \cite{Agresti} for the concrete $2\times 2$ identification of the interaction with the log-odds ratio; for the broader log-linear and graphical-model perspective, see also \cite{Lauritzen}.

\subsection{Monotone likelihood ratio}

For the binary observation $Y\in\{0,1\}$, define the likelihood ratio \[ \Lambda(v) \coloneqq \frac{\bbP(V{=}v \mid A)}{\bbP(V{=}v \mid \neg A)}. \] Then $\bbP(B\mid A) > \bbP(B\mid \neg A)$ is equivalent to $\Lambda(1) > \Lambda(0)$, since $\Lambda(1) = \frac{p/(p+q)}{r/(r+s)}$ and $\Lambda(0) = \frac{q/(p+q)}{s/(r+s)}$, so $\Lambda(1)/\Lambda(0) = ps/qr$ is the odds ratio. By Bayes' rule in odds form, $\Lambda(1) > \Lambda(0)$ means that observing $Y{=}1$ multiplies the odds of $A$ by a larger factor than observing $Y{=}0$ does, which is $\bbP(A\mid B) > \bbP(A\mid \neg B)$. This Neyman--Pearson-flavored argument \cite[Ch.~1]{Karlin} is technically similar to the Bayes factor proof; however, it is conceptually distinct because it frames the result as a monotone likelihood ratio property of the family $\{\bbP(\cdot \mid A),\, \bbP(\cdot \mid \neg A)\}$.

\subsection{Logistic regression}

In a $2\times 2$ table, the saturated logistic regression of $\mathbf{1}_A$ on $\mathbf{1}_B$ gives \[ \logit \bbP\parent{X{=}1 \mid Y{=}b} = \hat\beta_0 + \hat\beta_1\, b, \qquad b\in\{0,1\}, \] with $\hat\beta_1 = \log(ps/qr)$ \cite[on $2\times2$ tables, odds ratios, and logistic regression]{Agresti}. Fitting the reversed regression of $\mathbf{1}_B$ on $\mathbf{1}_A$ gives the same coefficient, since $\hat\beta_1$ is the log-odds ratio in both cases. Hence \[ \bbP(A\mid B) > \bbP(A\mid \neg B) \iff \hat\beta_1 > 0 \iff \bbP(B\mid A) > \bbP(B\mid \neg A). \] This is the odds-ratio proof dressed in the language of generalized linear models.

%% ====================================================================
\section{The Information-Theoretic Proof}\label{sec:pmi}
%% ====================================================================

\subsection{Perspective}

This section recasts \Cref{thmmain} in information-theoretic terms. The key point is that positive evidential support is equivalent to positive pointwise mutual information, whose symmetry in $A$ and $B$ makes the reciprocity of \Cref{thmmain} immediate.

\subsection{Proof}

\begin{definition}[Pointwise mutual information] The \defi{pointwise mutual information} (PMI) of the outcome pair $(A,B)$ is \cite[Ch.~2, esp.\ §2.3]{CoverThomas} \[ \pmi\parent{A;B} \coloneqq \log \frac{\bbP(A \wedge B)}{\bbP(A)\,\bbP(B)} = \log \frac{\bbP(A\mid B)}{\bbP(A)} = \log \frac{\bbP(B\mid A)}{\bbP(B)}. \] The three expressions are equal by Bayes' rule, and the first makes manifest the symmetry $\pmi\parent{A;B} = \pmi\parent{B;A}$. \end{definition}

\begin{lemma}\label{lemmapmi} We have $\bbP(A\mid B) > \bbP(A\mid \neg B)$ if and only if $\pmi\parent{A;B} > 0$. \end{lemma}

\begin{proof} By the law of total probability, we compute \[ \bbP(A) = \bbP(A\mid B)\,\bbP(B) + \bbP(A\mid \neg B)\,\bbP(\neg B). \] Since $\bbP(B)$ and $\bbP(\neg B)$ are strictly positive and sum to~$1$, the marginal $\bbP(A)$ is a strict convex combination of $\bbP(A\mid B)$ and $\bbP(A\mid \neg B)$. It therefore lies strictly between them unless they are equal. Hence \[ \bbP(A\mid B) > \bbP(A\mid \neg B) \iff \bbP(A\mid B) > \bbP(A). \] Since $\log$ is strictly increasing, $\bbP(A\mid B) > \bbP(A) \iff \pmi\parent{A;B} > 0$. \end{proof}

\begin{proof}[Proof of \Cref{thmmain}, information-theoretic version] By \Cref{lemmapmi}, we have \[ \bbP(A\mid B) > \bbP(A\mid \neg B) \iff \pmi\parent{A;B} > 0. \] By symmetry, \[ \pmi\parent{A;B} > 0 \iff \pmi\parent{B;A} > 0. \] By \Cref{lemmapmi} with $A$ and $B$ swapped, \[ \pmi\parent{B;A} > 0 \iff \bbP(B\mid A) > \bbP(B\mid \neg A). \] \end{proof}

\begin{remark}[Carnap's probability ratio] This proof is closely related to what the Bayesian confirmation literature calls the \defi{probability ratio rule} \cite{EellsFitelson,CrupiTentoriGonzalez}, namely $\bbP(A\mid B)/\bbP(A) = \bbP(B\mid A)/\bbP(B)$. This is pointwise mutual information with the logarithm removed. We do not treat it as a separate proof, but it deserves mention as the historical antecedent. \end{remark}

\subsection{Surprisal reduction}

The \defi{surprisal} of $A$ is $-\log \bbP(A)$; after learning $B$ it becomes $-\log \bbP(A\mid B)$. The reduction in surprisal is \[ \log \bbP(A\mid B) - \log \bbP(A) = \pmi\parent{A;B}. \] A positive value means $A$ became less surprising after learning $B$. By symmetry of PMI, the same quantity measures how much $A$ reduces the surprisal of $B$. The reciprocity of evidence is the reciprocity of surprisal reduction.

\subsection{Mutual information and the direction of dependence}

The ordinary mutual information
\[
I(A;B)\coloneqq \bbE[\pmi\parent{A;B}]
\]
is always nonnegative \cite[Ch.~2]{CoverThomas}, and it vanishes if and only if $A$ and $B$ are independent. It measures the amount of dependence, but not its direction: positively associated and negatively associated binary pairs can have the same mutual information.

The direction of association is instead visible at the level of pointwise mutual information for the four atomic outcome pairs. In our notation,
\[
\pmi(A;B)=\log\frac{p}{(p+q)(p+r)},\qquad
\pmi(\neg A;\neg B)=\log\frac{s}{(r+s)(q+s)},
\]
while
\[
\pmi(A;\neg B)=\log\frac{q}{(p+q)(q+s)},\qquad
\pmi(\neg A;B)=\log\frac{r}{(r+s)(p+r)}.
\]
A direct calculation shows that
\[
\sgn\pmi(A;B)=\sgn\pmi(\neg A;\neg B)=\sgn\Delta
\]
and
\[
\sgn\pmi(A;\neg B)=\sgn\pmi(\neg A;B)=-\sgn\Delta,
\]
where $\Delta=ps-qr$ as usual. Thus, for binary variables, the two concordant atoms $A\wedge B$ and $\neg A\wedge\neg B$ share one pointwise mutual information sign, the two discordant atoms $A\wedge\neg B$ and $\neg A\wedge B$ share the opposite sign, and the same scalar $\Delta$ governs them all.

%% ====================================================================
\section{The Combinatorial Proof}\label{sec:comb}
%% ====================================================================

\subsection{Perspective}

This section gives an elementary counting proof of \Cref{thmmain}. Instead of manipulating probabilities or transforms of probabilities, it compares favorable and unfavorable ordered pairs and shows that the same product inequality appears in both directions.

\subsection{Proof}

Assume the probability space is finite and uniform, with subsets \[ \mathcal{P} \coloneqq A\cap B,\quad \mathcal{Q} \coloneqq A\cap \neg B,\quad \mathcal{R} \coloneqq \neg A\cap B,\quad \mathcal{S} \coloneqq \neg A\cap \neg B, \] having cardinalities $|\mathcal{P}|=n_p$, $|\mathcal{Q}|=n_q$, $|\mathcal{R}|=n_r$, $|\mathcal{S}|=n_s$.

\begin{proof}[Proof of \Cref{thmmain}, combinatorial version] To compare $\bbP(A\mid B)$ and $\bbP(A\mid \neg B)$, consider ordered pairs $(u,v)$ with $u\in B$ and $v\in \neg B$. There are $(n_p + n_r)(n_q + n_s)$ such pairs. Call a pair \defi{$A$-favorable} if $u\in A$ and $v\notin A$, and \defi{$A$-unfavorable} if $u\notin A$ and $v\in A$. The $A$-favorable pairs are exactly those with $u\in \mathcal{P}$ and $v\in \mathcal{S}$, so there are $n_p n_s$ of them. The $A$-unfavorable pairs have $u\in \mathcal{R}$ and $v\in \mathcal{Q}$, so there are $n_r n_q$ of them. Therefore \[ \bbP(A\mid B) > \bbP(A\mid \neg B) \iff n_p n_s > n_r n_q. \]

Now swap the roles of $A$ and $B$. To compare $\bbP(B\mid A)$ and $\bbP(B\mid \neg A)$, consider ordered pairs $(u,v)$ with $u\in A$ and $v\in \neg A$. The $B$-favorable pairs have $u\in \mathcal{P}$ and $v\in \mathcal{S}$, again counted by $n_p n_s$. The $B$-unfavorable pairs have $u\in \mathcal{Q}$ and $v\in \mathcal{R}$, counted by $n_q n_r$. Hence \[ \bbP(B\mid A) > \bbP(B\mid \neg A) \iff n_p n_s > n_q n_r. \] The two inequalities are identical. \end{proof}

\subsection{Extension to arbitrary probability spaces}

The combinatorial proof works directly for finite uniform spaces. For rational probabilities, clear denominators to reduce to the finite uniform case. For real-valued tables with all entries strictly positive, the two conditional-probability differences \begin{align*} F(p,q,r,s) &\coloneqq \bbP(A\mid B) - \bbP(A\mid \neg B) \ \text{and} \\ G(p,q,r,s) &\coloneqq \bbP(B\mid A) - \bbP(B\mid \neg A) \end{align*} are continuous functions on the open simplex $\{(p,q,r,s)\in(0,1)^4 \mid p+q+r+s=1\}$. Since $\sgn F = \sgn G$ on the dense subset of rational tables, and both functions are continuous, the equality of signs extends to the entire simplex.

\subsection{The random-matchup interpretation}

The combinatorial proof has a probabilistic reformulation. Sample independently one point $u$ uniformly from $B$ and one point $v$ uniformly from $\neg B$. Then $\bbP(u\in A) = \bbP(A\mid B)$ and $\bbP(v\in A) = \bbP(A\mid \neg B)$, so $\bbP(A\mid B) > \bbP(A\mid \neg B)$ says the $B$-point is more likely to satisfy $A$ than the $\neg B$-point. The favorable-pair count $n_p n_s$ and the unfavorable-pair count $n_r n_q$ are rescalings of the joint probabilities of the events ``$u$ satisfies $A$ and $v$ does not'' and ``$u$ does not satisfy $A$ and $v$ does''---a natural tournament-style comparison.

%% ====================================================================
\section{The Log-Potential Proof}\label{sec:logpot}
%% ====================================================================

\subsection{Perspective}

This section recasts \Cref{thmmain} in terms of a discrete potential on $\{0,1\}^2$. The key point is that both conditional inequalities are equivalent to positivity of the same mixed discrete derivative of $\log P$, so the symmetry comes from the commutativity of discrete differentiation.

\subsection{Proof}

Write $g(a,b) \coloneqq \log \bbP(A{=}a,\, B{=}b)$ for $(a,b)\in\{0,1\}^2$, so \[ g(1,1) = \log p,\quad g(1,0) = \log q,\quad g(0,1) = \log r,\quad g(0,0) = \log s. \]

Define the \defi{horizontal log-odds field} \[ H(b) \coloneqq g(1,b) - g(0,b) = \log \frac{\bbP(A{=}1 \mid B{=}b)}{\bbP(A{=}0 \mid B{=}b)}, \] i.e., $H(b)$ is the log-odds of $A$ at level $b$ of $B$. Since $\logit : \pi \mapsto \log\parent{\frac{\pi}{1 - \pi}}$ is strictly increasing, \[ \bbP(A\mid B) > \bbP(A\mid \neg B) \iff H(1) > H(0). \]

Define the \defi{vertical log-odds field} \[ V(a) \coloneqq g(a,1) - g(a,0) = \log \frac{\bbP(B{=}1 \mid A{=}a)}{\bbP(B{=}0 \mid A{=}a)}. \] Then \[ \bbP(B\mid A) > \bbP(B\mid \neg A) \iff V(1) > V(0). \]

\begin{proof}[Proof of \Cref{thmmain}, log-potential version] The mixed discrete derivative of $g$ is \begin{align*} H(1) - H(0) &= \bigl[g(1,1) - g(0,1)\bigr] - \bigl[g(1,0) - g(0,0)\bigr] \\ &= g(1,1) + g(0,0) - g(1,0) - g(0,1). \end{align*} The same quantity arises from $V$, \begin{align*} V(1) - V(0) &= \bigl[g(1,1) - g(1,0)\bigr] - \bigl[g(0,1) - g(0,0)\bigr] \\ &= g(1,1) + g(0,0) - g(1,0) - g(0,1). \end{align*} The mixed discrete derivative $\partial_{AB} g \coloneqq g(1,1) + g(0,0) - g(1,0) - g(0,1)$ does not depend on the order in which the differences are taken. (This is the discrete analogue of $\partial_a \partial_b = \partial_b \partial_a$.) Since the logit function is strictly increasing, $H(1) > H(0)$ and $V(1) > V(0)$ are both equivalent to $\partial_{AB} g > 0$. \end{proof}

Note that $\partial_{AB} g = \log(ps/qr)$, so $\partial_{AB} g > 0$ if and only if $\OR > 1$ if and only if $\Delta > 0$. But the proof's engine is different from the odds-ratio proof---the invariant appears not as a ratio but as a mixed increment of a potential, and the symmetry comes from the commutativity of discrete differentiation rather than from any algebraic manipulation of odds.

\subsection{Total positivity and the FKG inequality}

The condition $\partial_{AB} g \ge 0$---equivalently $ps \ge qr$---is the $2\times 2$ case of \defi{total positivity of order~2} (TP${}_2$) \cite[Ch.~1]{Karlin}, or equivalently \defi{log-supermodularity}. A nonnegative matrix is TP${}_2$ if every $2\times 2$ minor is nonnegative. For a $2\times 2$ table there is only one minor, so TP${}_2$ reduces to $ps \ge qr$.

Log-supermodularity has remarkable consequences. The celebrated \defi{FKG inequality} \cite{FKG} states that if the joint law on a finite distributive lattice is log-supermodular, then any two increasing events are positively correlated. In the $2\times 2$ case, this is exactly \Cref{thmmain}. Total positivity and log-supermodularity are both invariant under transposition of the matrix, so the symmetry of \Cref{thmmain} is a special case of a general phenomenon.

\begin{remark} In the continuous setting, the analogue of log-supermodularity for a joint density $f(a,b)$ is $\partial^2 \log f / \partial a\,\partial b \ge 0$. The mixed discrete derivative $\partial_{AB} g$ is the discrete version of this condition. \end{remark}

\subsection{Stochastic order}

View the two conditional laws $\mu_A \coloneqq \mathcal{L}(B \mid A)$ and $\mu_{\neg A} \coloneqq \mathcal{L}(B \mid \neg A)$ as probability measures on the ordered two-point space $\{0,1\}$. Then $\bbP(B\mid A) > \bbP(B\mid \neg A)$ is exactly the statement $\mu_A \succ_{\text{st}} \mu_{\neg A}$ in the usual stochastic order \cite[Ch.~1]{ShakedShanthikumar}. Likewise, $\bbP(A\mid B) > \bbP(A\mid \neg B)$ says the analogous stochastic domination holds after swapping coordinates. In a $2\times 2$ table, the signed difference of the two conditional laws has only one degree of freedom, so row-wise and column-wise stochastic order are equivalent. For $m\times n$ tables with $\min(m,n) \ge 3$, these notions can diverge, which is why the simple symmetry fails in higher dimensions.

\subsection{The log-linear model}

The mixed increment $\partial_{AB} g = \log(ps/qr)$ is exactly the interaction term in the saturated log-linear model for a $2\times 2$ table. In log-linear form, the joint mass function may be written as \[ \log \bbP(A{=}a,\,B{=}b) = \mu + \lambda_A^{(a)} + \lambda_B^{(b)} + \lambda_{AB}^{(a,b)}, \] subject to the usual identifiability convention that the interaction terms sum to zero over each index. In the $2\times 2$ case, the interaction has one free parameter, and under the standard effect-coding convention that parameter is $\tfrac14\log(ps/qr)$. Thus the log-potential proof identifies \Cref{thmmain} with positivity of the unique interaction parameter in the saturated log-linear model \cite{Agresti}.

%% ====================================================================
\section{Conclusion}\label{sec:conclusion}
%% ====================================================================

We have given seven different proofs of \Cref{thmmain}.

\begin{enumerate} \item The \textbf{algebraic proof} cross-multiplies conditional-probability fractions. \item The \textbf{covariance proof} anchors each conditional probability to the marginal via $\bbP(A\mid B) - \bbP(A) = \Cov/\bbP(B)$ and appeals to $\Cov\parent{X,Y} = \Cov\parent{Y,X}$. \item The \textbf{coupling-parameter proof} expresses both conditional-probability differences as $t/(\text{positive constant})$, where $t$ is the unique free parameter once marginals are fixed. \item The \textbf{odds-ratio proof} passes to odds, reducing to $\OR > 1$. \item The \textbf{information-theoretic proof} uses the definitional symmetry $\pmi\parent{A;B} = \pmi\parent{B;A}$ and the law of total probability, never mentioning $\Delta$. \item The \textbf{combinatorial proof} counts favorable and unfavorable ordered pairs under two different decompositions. \item The \textbf{log-potential proof} identifies both inequalities with the positivity of the mixed discrete derivative of $\log P$, relying on the commutativity of discrete differentiation. \end{enumerate}

These proofs come from meaningfully different viewpoints. The algebraic, covariance, and odds-ratio proofs arrive at the scalar $\Delta = ps - qr$ (or its monotone transforms) by different routes---direct fraction manipulation, anchoring to the marginal, and passage to odds. The coupling-parameter proof works instead with $t = p - \bbP(A)\bbP(B)$, the unique interaction coordinate once marginals are fixed. The information-theoretic proof avoids $\Delta$ entirely. The combinatorial proof counts objects. The log-potential proof works on a different scale---$\log P$ rather than $P$---and derives the symmetry from the commutativity of a discrete operator rather than from an algebraic manipulation of fractions.

Around these proofs we developed a web of perspectives. The following table summarizes our main quantities of interest and what their positivity signifies.

\bigskip

\begin{center}
\renewcommand{\arraystretch}{1.3}
\begin{tabular}{lll}
\hline
\textbf{Perspective} & \textbf{Quantity} & \textbf{Positive means\ldots} \\
\hline
Linear algebra & $\det M = \Delta$ & rank exceeds $1$ \\
Geometry & signed area of row vectors & favorable tilt \\
Probability & $\Cov(\mathbf{1}_A,\mathbf{1}_B)$ & positive correlation \\
Regression & slope $\beta_{X\mid Y}$ & increasing conditional mean \\
Walsh--Fourier & interaction $\hat{f}_{UV}$ & positive interaction \\
Hilbert space & canonical correlation $\rho$ & aligned projections \\
Coupling & $t = p - \bbP(A)\bbP(B)$ & above independence \\
Concordance & $\bbP(\text{conc.}) > \bbP(\text{disc.})$ & aligned pairs favored \\
Transport & $\bbP(X\neq Y)$ & less mismatch than indep. \\
Statistics & $\log(ps/qr)$ & odds ratio $> 1$ \\
Bayesian & Bayes factor $> 1$ & $B$ supports $A$ \\
Ising model & interaction $\gamma$ & ferromagnetic coupling \\
MLR & $\Lambda(1)/\Lambda(0)$ & monotone likelihood ratio \\
Logistic & $\hat\beta_1 > 0$ & positive logistic slope \\
Information & $\pmi\parent{A;B}$ & positive local evidence \\
Combinatorial & $n_pn_s - n_qn_r$ & favorable ordered pairs dominate \\
Log-potential & $\partial_{AB}\log P$ & positive mixed increment \\
Total positivity & TP${}_2$ minor & positive association \\
Stochastic order & $\mu_A \succ_{\text{st}} \mu_{\neg A}$ & row-wise dominance \\
Log-linear & $\lambda_{AB} > 0$ & positive interaction term \\
\hline
\end{tabular}
\end{center}

\bigskip

All scalar association measures collapse to one parameter in the $2\times 2$ case, because the space of possible tables is a one-parameter family.

\begin{proposition} The four entries $p,q,r,s$ satisfy three independent linear constraints ($p+q = \bbP(A)$, $p+r = \bbP(B)$, $p+q+r+s = 1$), so $p$ is the only free parameter. Every signed measure of binary association that depends only on the joint law is therefore a function of $p$, or equivalently of $\Delta$, or equivalently of the coupling parameter $t$. \end{proposition}

This one-dimensionality explains why all nineteen perspectives in the table above agree on the sign of association: they are different coordinate systems on the same one-dimensional family. For $m\times n$ tables with $\min(m,n) \ge 3$, the interaction structure becomes multi-dimensional, and a single scalar no longer governs every reasonable notion of association; for example, the Hilbert-space canonical correlations become a vector, the odds-ratio generalizes to a matrix of local odds ratios, and row-wise versus column-wise stochastic order can disagree. The binary case is special, because a single number can tell the whole story.

\end{document}